\documentclass[a4paper, 11pt]{article}

\usepackage{a4wide}
\usepackage{amsmath}

\usepackage{amsthm}
\usepackage{commath}
\usepackage{latexsym}
\usepackage{graphicx}
\usepackage{color}
\usepackage[pdfstartview=FitH]{hyperref}
\usepackage{verbatim}
\usepackage{authblk}
\usepackage[section]{placeins}

\title{Characterisation of the $\chi$-index and the $rec$-index}

\author[1]{Mark Levene mark@dcs.bbk.ac.uk}
\author[1]{Trevor Fenner trevor@dcs.bbk.ac.uk}
\author[2]{\authorcr Judit Bar-Ilan  Judit.Bar-Ilan@biu.ac.il}
\affil[1]{Department of Computer Science and Information Systems, \authorcr University of London, London WC1E 7HX, U.K.}
\affil[2]{Department of Information Science, Bar-Ilan University, Ramat Gan, Israel}

\date{}

\begin{document}

\maketitle

\newtheorem{theorem}{Theorem}[section]
\newtheorem{corollary}[theorem]{Corollary}
\newtheorem{proposition}[theorem]{Proposition}

\begin{abstract}

Axiomatic characterisation of a bibliometric index provides insight into the properties that the index satisfies and facilitates the comparison of different indices.
A geometric generalisation of the $h$-index, called the $\chi$-index, has recently been proposed to address some of the problems with the $h$-index, in particular, the fact that it is not scale invariant, i.e.,
multiplying the number of citations of each publication by a positive constant may change the relative ranking of two researchers.
While the square of the $h$-index is the area of the largest square under the citation curve of a researcher, the square of the $\chi$-index, which we call the $rec$-index (or {\em rectangle}-index), is the area of the largest rectangle under the citation curve. Our main contribution here is to provide a characterisation of the $rec$-index via three properties: {\em monotonicity}, {\em uniform citation} and {\em uniform equivalence}. Monotonicity is a natural property that we would expect any bibliometric index to satisfy, while the other two properties constrain the value of the $rec$-index to be the area of the largest rectangle under the citation curve. The $rec$-index also allows us to distinguish between {\em influential} researchers who have relatively few, but highly-cited, publications and {\em prolific} researchers who have many, but less-cited, publications.

\end{abstract}


\section{Introduction}
\label{sec:introduction}

Axiomatic characterisation of a bibliometric index \cite{MARC09,BOUY14} provides insight into the properties that the index satisfies and facilitates the comparison of different indices. (Axiomatic characterisation is used in a number of areas for the same purpose.)

\smallskip

Since Garfield's foundational work in bibliometrics \cite{GARF79}, a plethora of bibliometric indices have been suggested \cite{ROEM15,TODE16}, from a simple count of the total number of citations to the more sophisticated $h$-index \cite{HIRS05}.
These have often been motivated by the ongoing debate between quantity (as measured by the number of publications) and quality (as measured by the number of citations to those publications).
A review of many of these \cite{WILD14}, where a comparison of 108 bibliometric indicators was presented, concluded that, in order to gauge the overall impact of a researcher, several indicators should be used. In particular, many variants of the $h$-index have been proposed and, in \cite{BORN11}, it was shown that generally there is high correlation between the $h$-index and 37 of its variants. Moreover, a critical view of the $h$-index and its variants, which was provided in \cite{SCHR18}, argued that it is not a good indicator of recent impact and suggested instead a timed $h$-index over a sliding time window.
Another criticism of the $h$-index and its variants is that it treats all citations in the same way, ignoring the relevance and impact of citations. This issue can be addressed by Markovian methods, such as PageRank \cite{LEVE10}, giving rise to indices based on author or publication-level citation networks, such as the Eigenfactor score \cite{WEST13}, which can model influence between researchers and publications.
However, Markovian indicators have the disadvantage of being significantly more complex to compute than those based on citation counts, and
there is no conclusive evidence that they actually outperform rankings based on citation counts \cite{FIAL15}.

\smallskip

We mention two other issues concerning bibliometic indices that have been recently addressed in a formal setting.
The first is that of normalising citation counts across different fields. For example, in \cite{BOUY16}, the authors investigated fractional counting of citations, where the value of a citation is inversely proportional to the numbers papers being cited. They provide a characterisation of the ranking induced by such fractional counting. The second issue is aggregation of bibliometric indices when there are several conflicting indices and there is no compelling reason to chose one over another. For example, in \cite{SUBO18}, the authors investigated ranking methods from social choice theory in order to provide an axiomatic analysis of aggregation of bibliometric indices.

\smallskip

In this paper we focus on a (two-dimensional) geometric generalisation of the $h$-index, called the $\chi$-index \cite{FENN18b}, which has recently been proposed to address some of the problems with the $h$-index, in particular the fact that it is not {\em scale invariant} \cite{PERR16}.

\smallskip

The $h$-index is determined by the largest square that fits under the citation curve of a researcher when plotting the number of citations to individual publications in decreasing order. On the other hand, the $\chi$-index is determined by the largest area rectangle that fits under the citation curve.
The $rec$-index (or {\em rectangle}-index) is defined to be the area of this rectangle and is the square of the $\chi$-index.
The $rec$-index is thus a member of the class of geometric indices that approximate the area under the citation curve of a researcher. Such indices can address problems attached to, for example, the citation count, which takes into account
all citations. The $h$-index penalises both highly-cited publications and publications with only a few citations.
On the other hand, the $\chi$-index is more balanced than the $h$-index in this respect, as it allows us to cater for both {\em influential} researchers with a few very highly-cited publications and {\em prolific} researchers who may have many publications but relatively few citations per publication.

\smallskip

Our main contribution here is to provide a first characterisation of the $rec$-index via three properties: {\em monotonicity}, {\em uniform citation} and {\em uniform equivalence}. Monotonicity is a natural property that we would expect any bibliometric index to satisfy, while the other two properties relate to the rectangle under the citation curve that determines the index. Uniform citation specifies that when the shape of the citation curve is uniform, i.e. rectangular, then the value of the index is the total citation count of all the publications. Complementing this property, uniform equivalence specifies that when the shape of the citation curve is not uniform, i.e. not rectangular, the value of the index is equal to that of some uniform citation curve that can obtained by omitting some number of citations.

\smallskip

We note that many different properties may be used to characterise a bibliometric index, and there is no general agreement on which are the most compelling. For example, there are a number of distinct characterisations of the $h$-index, such as those presented in \cite{WOEG08} and \cite{QUES11}.
However, if a number of properties are proved to characterise a given index, the acceptance of any other index would necessitate the violation of at least
one of the properties used to characterise the given index.

\medskip

The rest of the paper is organised as follows.
In Section~\ref{sec:chi-index}, we introduce the $rec$-index.
In Section~\ref{sec:axioms}, we define and discuss properties of the $rec$-index and other indices.
Then, in Section~\ref{sec:proof}, we present an axiomatic characterisation of the $rec$-index.
Finally, in Section~\ref{sec:conc}, we give our concluding remarks.

\section{The $rec$-index and related bibliometric indices}
\label{sec:chi-index}

We assume that a researcher publishes $n$ publications, where $n \ge 0$, which are represented by a {\em citation vector} of {\em positive} integers, ${\bf x} = \left<x_1, x_2, \ldots, x_n\right>$, where $x_i$ is the number of citations to publication $i$, sorted in descending order, i.e. $x_i \ge x_j$ for $1 \le i < j \le n$. (We note that one could consider only a subset of a researcher's publications in the citation vector, for example, by only allowing journal publications, or publications in high-impact venues.)

\smallskip

The {\em citation curve} is the curve arising from plotting the number of citations against the ranking of the publications as a histogram specified by the citation vector.

\smallskip

A {\em bibliometric index} is a function $f$ that maps citation vectors to the set of non-negative real numbers.
As in \cite{WOEG08}, we assume the {\em baseline condition} that, for the empty citation vector ${\bf x} = \left<\right>$,
we have $f({\bf x})=0$.

\smallskip

In this paper, we concentrate on characterising the $rec$-index directly, following the approach adopted in \cite{WOEG08,WOEG08a,QUES11,QUES11a}, rather than characterising the bibliometric ranking induced by the index, as was done in \cite{MARC09,BOUY14}.
This stems from our particular interest in the properties of geometric indices.
It is evident that any two indices, such as the $\chi$-index and the $rec$-index, that are monotonic transformations of each other are equivalent with respect to the induced rankings.

\smallskip

The {\em citation count} index for a citation vector ${\bf x} =  \left<x_1,x_2,\ldots,x_n\right>$ is the L1 norm of ${\bf x}$, denoted by $\lVert {\bf x} \rVert$ and defined by
\begin{equation}\label{eq:cit}
\lVert {\bf x} \rVert = \sum_{i=1}^n x_i.
\end{equation}
\smallskip

We say that a citation vector ${\bf x} = \left<x_1,x_2,\ldots,x_n\right>$ is {\em dominated} by a citation vector ${\bf y} = \left<y_1,y_2,\ldots,y_m\right>$,
written as ${\bf x} \sqsubseteq {\bf y}$, if $n \le m$ and $x_i \le y_i$ for all $i$, $1 \le i \le n$.

\medskip

It is worth recalling some known bibliometric indices \cite{WILD14,TODE16}: the citation count, as defined in (\ref{eq:cit}); the publication count $n$;
the maximum citation index $x_1$; the Euclidean index $E({\bf x})$ \cite{PERR16}, which is the Euclidean norm of ${\bf x}$, i.e., $E({\bf x}) = \sqrt{\sum_{i=1}^{n} x_i^2}$; the $h$-index \cite{HIRS05}; and the $g$-index \cite{EGGH06}, which is a variant of the $h$-index giving extra weight to highly-cited publications.

\smallskip

The $h$-index \cite{HIRS05}, in particular, has gained popularity due to its relative simplicity, ease of calculation, and its ingenious method of combining the quality and quantity of a researcher's outputs. It is defined as the maximum number $h$ of the researcher's publications such that each has at least $h$ citations, i.e. for a citation vector, ${\bf x} = \left<x_1, x_2, \ldots, x_n\right>$ the $h$-index is the largest $h$ for which $x_h \ge h$.

\smallskip

To motivate the $\chi$-index, consider the following three citation vectors, the citation curves of which are depicted in Figure~\ref{figure:example}:
(i) $\left< 100 \right>$, i.e. 1 publication with 100 citations, (ii) $\left< 10,10, \cdots,10 \right>$, i.e. 10 publications with 10 citations each, and (iii) $\left< 1,1, \cdots, 1 \right>$, i.e. 100 publications with 1 citation each. (Note that the diagram in Figure~\ref{figure:example} is not drawn to scale.)

\begin{figure}[htb]
\centering{\includegraphics[scale=0.125]{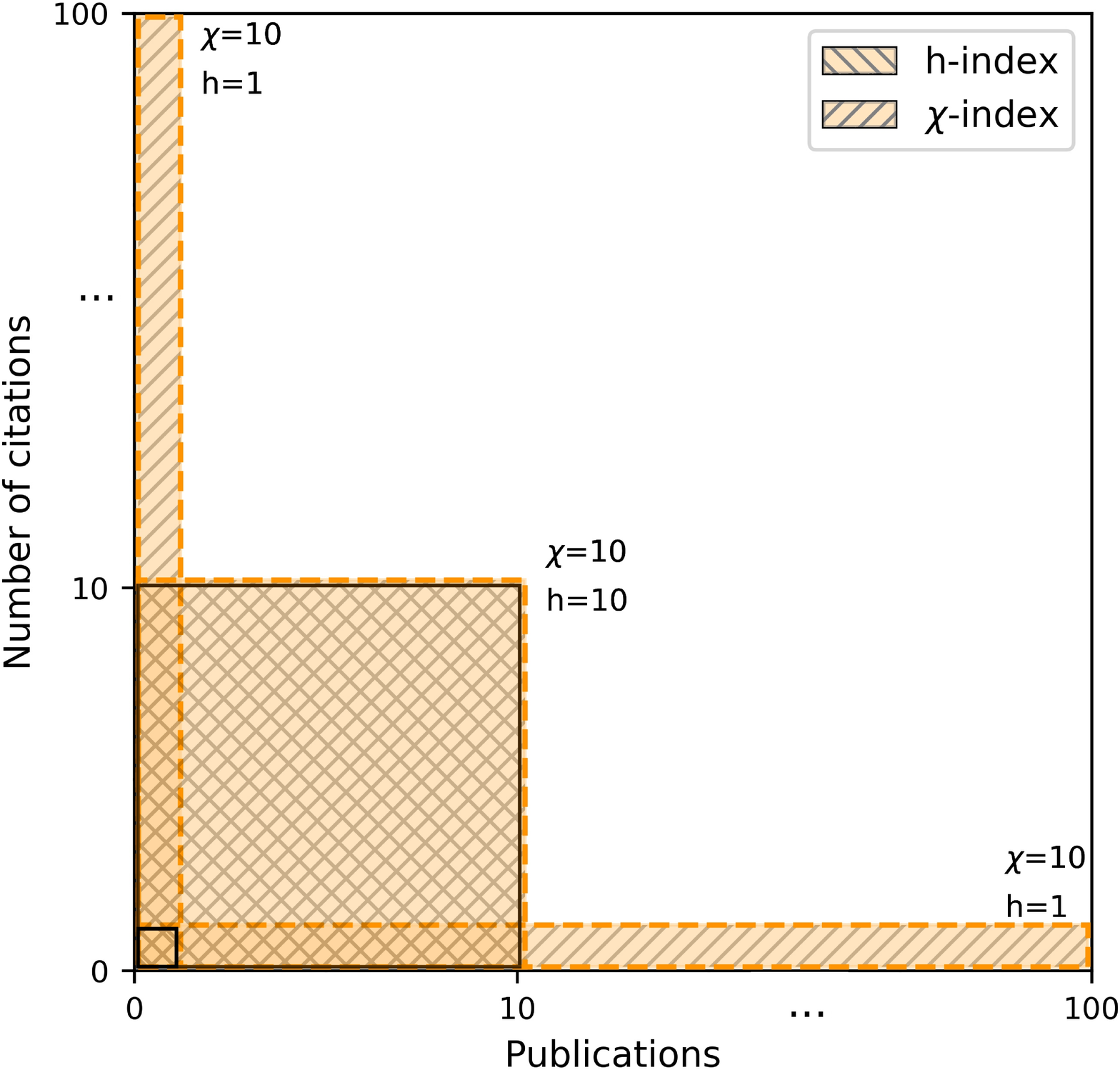}}
\caption{\label{figure:example} Example of the geometric interpretation of the $h$ and $\chi$ indices.}
\end{figure}
\smallskip

The $\chi$-index is defined as the square root of the maximum area rectangle that can fit under the citation curve, while the $h$-index is the square root of the maximum area square that can fit under the citation curve.

\smallskip

Formally, we first define the $rec$-index (or {\em rectangle}-index) of a researcher with citation vector ${\bf x}$ by
\begin{equation}\label{eq:ck-index}
rec({\bf x}) = \max_{i} i x_i.
\end{equation}
\smallskip

\noindent The $\chi$-index \cite{FENN18b} is then defined by $\chi({\bf x}) = \sqrt{rec({\bf x})}$.

\medskip

Returning to our example shown in Figure~\ref{figure:example}, we see that all three researchers have a $\chi$-index of $10$, while researcher (ii) has an $h$-index of $10$, but researchers (i) and (iii) both have an $h$-index of only $1$. The $h$-index may be seen as balancing quality, on the one hand, by favouring publications with a higher number of citations and quantity, on the other hand, by taking into account all publications with a sufficient number of citations. However, such an approach disadvantages a researcher, such as (i), with a few very highly-cited publications, who may have carried out some influential seminal research, and it also disadvantages a prolific researcher, such as (iii), who may have many publications but fewer citations per publication. Now, if we let $k$ denote a value of $i$ that maximises $i x_i$ in (\ref{eq:ck-index}), the $rec$-index can distinguish between
more {\em influential} researchers for which $x_k > k$, such as (i), and  more {\em prolific} researchers for which $k > x_k$, such as (iii). In this sense the $\chi$-index avoids the debate of number of citations versus number of publications by awarding all three researchers the same $\chi$-index of $10$.

\medskip

The $rec$-index is a member of the class of (two-dimensional) geometric indices, as is the square of the $h$-index, and also the half the square of the $w$-index \cite{WOEG08}, which is the area of the maximal isosceles right-angled triangle under the citation curve; formally, the $w$-index is the largest integer $w$ such that the citation vector ${\bf x}$ contains $w$ distinct publications with at least $1, 2, \ldots, w$ citations, respectively.
In a more general setting, the dimension of an index can be formally defined  \cite{PRAT17}, and is related to the  requirement of {\em dimensional homogeneity} from physics, that we may only compare quantities that have the same units.

\smallskip

Geometric indices are actually quite natural, as their goal is to consider the area under the citation curve in order to encapsulate the essential citations for a set of {\em core} publications that in some sense represent the output of a researcher.
We note that the $rec$-index, the square of the $h$-index, the publication count, and the maximum citation index all include the same number of citations for each core publication.

\smallskip

The citation count, which includes all publications with at least one citation in the core, is of course a reasonable bibliometric index. However, it is often argued that the citation count is problematic; in particular, it may be inflated by a small number of publication having a large number of citations or it may be sensitive to a long tail of publications each having only few citations. The axiomatic characterisation we describe here can be viewed as contributing to this debate by discussing several characteristics of geometric indices.

\section{Some properties of bibliometric indices}
\label{sec:axioms}

In this section we define a variety of properties of bibliometric indices, almost all of which are satisfied by the $rec$-index, and then, in Section~\ref{sec:proof}, we show that {\em monotonicity} combined with {\em uniform citation} and {\em uniform equivalence} characterise the $rec$-index.

\smallskip

We assume throughout this section that $f(\cdot)$ is the index under consideration.
The first property, monotonicity, is a natural requirement for any bibliometric index, stating that adding citations to the citation vector should not decrease the value of the index:
\begin{enumerate}
\item[] {\bf Monotonicity (M):}
If citation vector ${\bf x}$ is dominated by citation vector ${\bf y}$, i.e. ${\bf x} \sqsubseteq {\bf y}$, then $f({\bf x}) \le f({\bf y})$.
\end{enumerate}
\smallskip

\noindent
It is easy to verify that all the indices we consider are monotonic.
We note that the citation count satisfies the stronger property of {\em strict monotonicity} {(\bf SM)}, in which $f({\bf x}) < f({\bf y})$ when ${\bf x} \ne {\bf y}$. It also satisfies several other desirable properties, such as {\em rank independence} (adding a new publication with a given number of citations to two citation vectors does not change their relative ranking) and {\em rank scale invariance} (multiplying the number of citations of each publication by a positive constant does not change the relative ranking of two citation vectors) \cite{PERR16}. Neither the $\chi$-index nor the $h$-index are rank independent. However, the $\chi$-index is rank scale invariant, but the $h$-index is not. Moreover, the $rec$-index satisfies the following stronger property of (linear) scale invariance:
\begin{enumerate}
\item[] {\bf Scale invariance (SI):}
$f(C {\bf x}) = C f({\bf x})$, for any positive constant $C$.
\end{enumerate}
\smallskip

\noindent
It is easy to see that {\em scale invariance} implies rank scale invariance.

\smallskip

A natural form of symmetry can be attained via the {\em conjugate partition} of a citation vector ${\bf x} = \left<x_1, x_2, \ldots, x_n\right>$, which is the {\em publication vector} ${\bf p} = \left<p_1, p_2, \ldots, p_m\right>$, where $m = x_1$ and $p_i$ is the number of publications with at least $i$ citations \cite{WOEG08}. Geometrically, the publication vector is obtained by reflecting the geometric representation of the citation vector along the main diagonal. This is shown in Figure~\ref{figure:cp} for a citation vector $\left<6,4,3,1\right>$, on the left, and its conjugate partition publication vector $\left<4,3,3,2,1,1\right>$, on the right. This motivates the following property.

\begin{figure}[htb]
\centering{\includegraphics[scale=0.55]{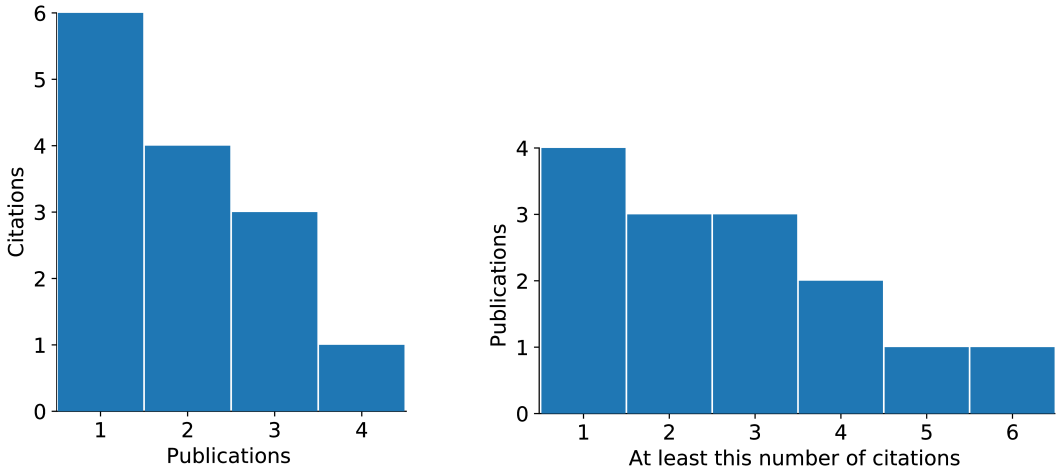}}
\caption{\label{figure:cp} The citation vector is shown on the left and its conjugate partition, the corresponding publication vector, is shown on the right.}
\end{figure}

\begin{enumerate}
\item[] {\bf Self-conjugacy (SC):}
Let ${\bf p}$ be the conjugate partition of the citation vector ${\bf x}$, then $f({\bf p}) = f({\bf x})$.
\end{enumerate}
\smallskip

Clearly the $rec$ and $h$ indices, as well as the citation count, are all self-conjugate.
On the other hand, we note that the publication count and the maximum citation index are conjugates of each other.
Self-conjugacy implies a balanced approach between influence (quality) and prolificity (quantity).

\smallskip

Some indices tend to emphasise influence, for example, the maximum citation index, Euclidean index and $g$-index, whereas others, such as the publication count, emphasise prolificity. Should we wish to emphasise influence rather than prolificity, we may define a version of the $rec$-index, the ${rec}_{_I}$-index, in which we restrict the maximum in ({\ref{eq:ck-index}) to be over those $i$ for which $i \le x_i$.
Conversely, should we wish to emphasise prolificity, we may instead define the conjugate index, the ${rec}_{_P}$-index, by using the corresponding publication vector ${\bf p}$  and restricting $i$ so that $i \le p_i$.

\smallskip

Many people, probably the majority, tend to favour indices that emphasise influence. It is therefore worth noting our findings in \cite{FENN18b}, where the citations of a large number of researchers, from a Google Scholar data set made available by Radicchi and Castellano \cite{RADI13}, were analysed and their $rec$-indices calculated. Table 11 in \cite{FENN18b} shows that $93\%$ of the researchers for which the $\chi$-index was significantly larger than the $h$-index were more influential than prolific, i.e. for these researchers $x_k > k$.
This indicates that, in general, the $rec$-index satisfies the tendency to favour influence.

\smallskip

Following the conclusion in \cite{WILD14} that several bibliometric indicators should be used to gauge the overall impact of a researcher, we suggest that both quality and quantity may be assessed by using the pair of indices $({rec}_{_I}, {rec}_{_P})$.

\smallskip

We now concentrate on the $rec$-index.
A typical citation curve, corresponding to the citation vector $\left<6,4,3,1\right>$, is shown in Figure~\ref{figure:fp}, with circles indicating where a new citation can be added. We observe that each addition of a new citation completes a rectangle.
If the new citation is added to publication $k$, then the newly formed rectangle has width $k$ and height $x_k +1$.
For example, if we add a citation to the third publication, producing the citation vector $\left<6,4,4,1\right>$, then the $rec$-index will increase from $9$ to $12$, but if we add it to the fourth publication, producing the citation vector $\left<6,4,3,2\right>$, then the $rec$-index will not increase.

\begin{figure}[htb]
\centering{\includegraphics[scale=0.08]{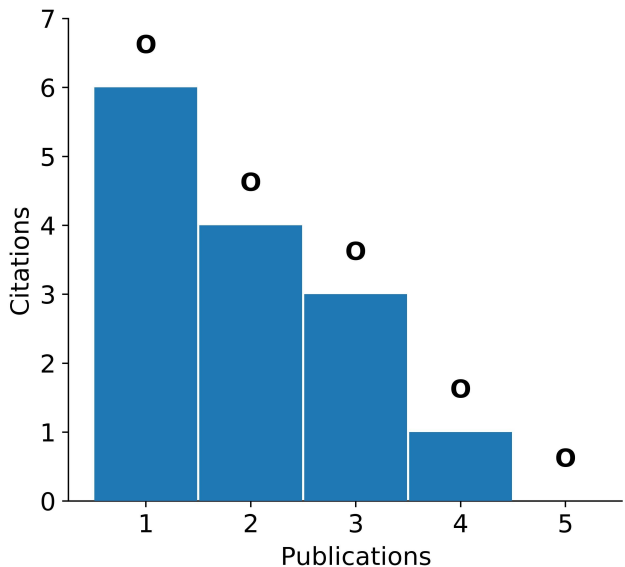}}
\caption{\label{figure:fp} A citation vector with circles indicating where a new citation may be added.}
\end{figure}
\smallskip

For any citation vector ${\bf x} = \left<x_1, x_2, \ldots, x_n\right>$, we write ${\bf x}^{[k]}$ for the citation vector obtained from ${\bf x}$ by adding a single citation to publication $k$, where $1 \le k \le n+1$, thereby increasing its citation count from $x_k$ to $x_k +1$. If $k = n+1$, we assume that $x_k = 0$. We note that $k$ must be the smallest index $j$ for which $x_j = x_k$.

\smallskip

By the definition of the $rec$-index, it is straightforward to see that $rec({\bf x}^{[k]}) \ge k (x_k+1)$, and thus
\begin{equation}\label{eq:rc}
rec({\bf x}^{[k]}) = max(rec({\bf x}), k (x_k+1)).
\end{equation}

\medskip

The following property encapsulates this observation.

\begin{enumerate}
\item[] {\bf Rectangle completion (RC):}
For any citation vector ${\bf x}$,
\begin{displaymath}
f({\bf x}^{[k]}) = max(f({\bf x}), k (x_k +1)).
\end{displaymath}
\end{enumerate}

Noting that {\em rectangle completion} implies {\em monotonicity}, it is then straightforward to verify that a bibliometric index $f$ satisfies {\em rectangle completion} (together with the baseline condition) if and only if it is the $rec$-index.
Obviously, as it essentially encapsulates the definition of the $rec$-index, {\em rectangle completion} is not particularly useful as a characterisation of the $rec$-index; however, it provides an alternative and constructive definition of the index.

\smallskip

Consider the situation when $rec({\bf x}) = k x_k$.
We note that (i) if $x_k = k$, the $rec$-index is equal to the square of the $h$-index, (ii) if $x_k > k$, the researcher tends towards being more influential, and (iii) if $x_k < k$, the researcher tends towards being more prolific. Thus the shape of the maximum area rectangle will indicate an interpretation of the index value. We also note that, when the histogram of the citation curve is a rectangle, the distribution of citations is {\em uniform}. In this case the {\em core} includes all publications.

\smallskip

More formally, we say that a citation vector ${\bf u} = \left<u_1,u_2,\ldots,u_n\right>$  is {\em uniform} if $u_1 = u_2 = \cdots = u_n$.
It follows that $rec({\bf x}) = \lVert {\bf x} \rVert$ if and only if ${\bf x}$ is uniform. This observation suggests the following weaker form of this property.

\begin{enumerate}
\item[] {\bf Uniform citation (UC):}
If the citation vector ${\bf x}$ is uniform then $f({\bf x}) = \lVert {\bf x} \rVert$.
\end{enumerate}
\smallskip

This property makes the reasonable assertion that, if all publications have the same number of citations, they should all be in the core and all citations to them included in the index.

\smallskip

As stated above {\em rectangle completion} is too contrived in the sense that it mimics the definition of the $rec$-index. So we now explore a way to replace it with {\em uniform citation}, which gives a lower bound on the index, together with another property that gives an upper bound on the index. This additional property is:
\begin{enumerate}
\item[] {\bf  Uniform equivalence (UE):}
For any citation vector ${\bf x}$, there exists a uniform citation vector ${\bf u}$ dominated by ${\bf x}$, i.e. ${\bf u} \sqsubseteq {\bf x}$,
for which $f({\bf x}) = f({\bf u})$.
\end{enumerate}
\smallskip

This property asserts that the same number of citations should be included for each publication in the core.

\smallskip

The following proposition, which defines a property that is similar to Axiom D in \cite{WOEG08}, can easily be shown to follow from {\em monotonicity}, {\em uniform citation} and {\em uniform equivalence}.

\begin{proposition}\label{prop:axiomd}
If $f$ satisfies properties {\bf M}, {\bf UC} and {\bf UE} then it also satisfies the following property:
\begin{enumerate}
\item[] {\bf  Citation increase (CI):}
If we add a single citation to each publication in ${\bf x}$, resulting in the citation vector ${\bf y}$, then $f({\bf y}) > f({\bf x})$.
\end{enumerate}
\end{proposition}
\smallskip

However, {\em citation increase} together with {\em monotonicity} and {\em uniform citation} do not imply {\em uniform equivalence}, as is the case for the bibliometric index,
\begin{equation}\label{eq:avg}
f({\bf x}) = \frac{rec({\bf x}) + \lVert {\bf x} \rVert}{2},
\end{equation}
which satisfies {\em citation increase} but not {\em uniform equivalence}.

\section{Axiomatic characterisation of the $rec$-index}
\label{sec:proof}

The rankings induced by several bibliometric indices, including the $h$-index and $g$-index, were characterised in \cite{MARC09,BOUY14},
whereas, in \cite{WOEG08,WOEG08a} and \cite{QUES11,QUES11a}, the authors concentrated on characterising the $h$-index and $g$-index directly.
These characterisations, as well as that presented here, are based on properties that address three types of issues.
First, the inclusion of fundamental properties like baseline and monotonicity should be considered.
The second issue is concerned with the conditions under which the value of the index increases (see, for example, {\em citation increase}).
The third issue considers what changes to the citation vector leave the index unchanged (for example, the $h$-index satisfies the property of {\em independence of irrelevant citations} \cite{QUES11a}, cf. \cite{KONG14}, which essentially states that adding a single citation to a core publication does not increase the index). Another category of properties that are important for characterising bibliometric indices are invariants, such as {\em scale invariance} and {\em self-conjugacy}, which give transformations that change the index in a predictable manner or do not change the value at all.

\smallskip

The main result in this section is Theorem~\ref{th:chi2}, which provides an axiomatic characterisation of the $rec$-index.
We also show in Proposition~\ref{prop:chi} that the $\chi$-index satisfies the desirable property that, subsequent to the addition of a single citation to the citation vector, the $\chi$-index cannot increase by more than one.

\begin{theorem}\label{th:chi2}
A bibliometric index $f$ satisfies the three properties of monotonicity, uniform citation and uniform equivalence if and only if it is the $rec$-index.
\end{theorem}
\smallskip

\begin{proof}
It is clear from Section~\ref{sec:axioms} that the $rec$-index satisfies these properties, so it remains to prove that they are sufficient.

\smallskip

Let ${\bf x}$ be a citation vector and let ${\bf u}$ be a uniform citation vector such that
$\lVert {\bf u} \rVert$ is maximal for all ${\bf u} \sqsubseteq {\bf x}$. Clearly, $rec({\bf x}) = \lVert {\bf u} \rVert$.
So, by {\em monotonicity} and {\em uniform citation}, we have $f({\bf x}) \ge f({\bf u}) = \lVert {\bf u} \rVert$.
By {\em uniform equivalence}, we may let ${\bf v}$ be a uniform citation vector such that ${\bf v} \sqsubseteq {\bf x}$ and $f({\bf x}) = f({\bf v})$,
and therefore $f({\bf x}) = \lVert {\bf v} \rVert$ by {\em uniform citation}. Since $\lVert {\bf u} \rVert \ge \lVert {\bf v} \rVert$ by the definition of ${\bf u}$, it follows that $f({\bf x}) = \lVert {\bf u} \rVert = rec({\bf x})$.
\end{proof}
\smallskip

The following two corollaries show that we may replace two of the properties in Theorem~\ref{th:chi2} by two simpler properties together with the intuitive property of {\em scale invariance}.

\begin{corollary}\label{cor:chi2}
The result of Theorem~\ref{th:chi2} holds if we replace monotonicity by the following more restrictive property.

\begin{enumerate}
\item[] {\bf  Uniform monotonicity (UM):}
If ${\bf x}$ is uniform and ${\bf x} \sqsubseteq {\bf y}$, then $f({\bf x}) \le f({\bf y})$.
\end{enumerate}
\smallskip
\end{corollary}

\begin{corollary}\label{cor:si}
The result of Theorem~\ref{th:chi2} holds if we replace uniform citation by the following more restrictive property together with scale invariance.

\begin{enumerate}
\item[] {\bf  Uniform single citation (USC):}
If  $x_1 = x_2 = \cdots = x_n = 1$ for the citation vector ${\bf x}$, then $f({\bf x}) = \lVert {\bf x} \rVert = n$.
\end{enumerate}
\end{corollary}
\smallskip

\begin{proof}
Clearly {\bf USC} and {\bf SI} imply {\bf UC}.
\end{proof}
\smallskip

It is not difficult to demonstrate that the three properties of Theorem~\ref{th:chi2} characterising the $rec$-index are independent, i.e. omitting any one of them would render the theorem false.

\begin{proposition}\label{prop:indep}
The three properties of Theorem~\ref{th:chi2} characterising the $rec$-index are independent.
\end{proposition}

\begin{proof}
The following examples justify this claim.
\renewcommand{\labelenumi}{\alph{enumi})}
\begin{enumerate}
\item The index defined in (\ref{eq:avg}) satisfies {\bf M} and {\bf UC} but not {\bf UE}.
\item The square of the $h$-index, the publication count $n$, the maximum citation index $x_1$, $max(n,x_1)$ and $min(n,x_1)$
all satisfy {\bf M} and {\bf UE} but not {\bf UC}.
\item The product $n x_n$ of the number of publications and the minimum number of citations satisfies {\bf UC} and {\bf UE} but not {\bf M}, nor
{\bf UM}.
\end{enumerate}
\end{proof}

To summarise the properties that characterise the $rec$-index: {\bf UM} or {\bf M} implies that adding new citations will not decrease the value of the index, while {\bf UC} and {\bf UE} provide, respectively, lower and upper bounds on its value.
We recall that, since the $rec$-index is not strictly monotonic, some citations may not contribute to the value of the index;
however, all core publications contribute an equal number of citations.

\medskip

We now present a construction that gives rise to a property equivalent to {\bf UE}.
While indices, such as the citation count (\ref{eq:cit}), that include the full set of citations are strictly monotonic, indices, such as the $rec$-index, that include just a core set of publications typically increase only after a batch of citations has been added. The construction we now present captures this aspect of the $rec$-index.

\smallskip

Consider a sequence of citation vectors
\begin{equation}\label{eq:s-seq}
{\bf S} \ = \ {\bf x}_1, {\bf x}_2, \ldots,  {\bf x}_i, \ldots, {\bf x}_s,
\end{equation}
where ${\bf x}_1 = \left<\right>$ and ${\bf x}_s = {\bf x}$.
When ${\bf x}_i \sqsubseteq {\bf x}_{i+1}$ and
$\lVert {\bf x}_{i+1} \rVert - \lVert {\bf x}_i \rVert =1$, for $1 \le i \le s-1$,
we say that ${\bf S}$ is a {\em constructive} sequence for ${\bf x}$.

\smallskip

We are interested in constructive sequences satisfying the property that, for each $i$, if $f({\bf x}_i) < f({\bf x}_{i+1})$ then ${\bf x}_{i+1}$ is a uniform citation vector; we call such a constructive sequence $f$-{\em incremental}.
This suggests the following property of a citation index.

\begin{enumerate}
\item[] {\bf  Uniform increment (UI):}
For any citation vector ${\bf x}$, there exists an $f$-{\em incremental} constructive sequence for ${\bf x}$.
\end{enumerate}
\smallskip

It is not difficult to prove that {\bf UI} implies {\bf UE}. Moreover, the $rec$-index satisfies {\bf UI}, and we now present one method for constructing a $rec$-{\em incremental} sequence for a citation vector ${\bf x}$.

\renewcommand{\labelenumi}{(\roman{enumi})}
\begin{enumerate}
\item Start from $\left<\right>$ and construct a sequence of uniform citation vectors as follows.
\item From the uniform citation vector ${\bf u}$, add citations one-by-one to obtain a new uniform citation vector, either by adding a new column (i.e. a new publication) or by adding a new row (cf. property {\bf CI}) to the rectangle corresponding to ${\bf u}$.
\item Repeat step (ii) until we obtain a citation vector ${\bf u}$ for which $rec({\bf u}) = rec({\bf x})$.
\item Add the remaining citations one-by-one in any order until we obtain ${\bf x}$.
\end{enumerate}
\smallskip

It is straightforward to show that it is always possible, in step (ii) above, to choose between adding a column or a row in such a way that the sequence is $rec$-{\em incremental}.

\medskip

It may be argued that it is natural that a one-dimensional bibliometric index should not increase by more than one when a single citation is added to the citation vector. We now prove that this holds for the $\chi$-index.

\begin{proposition}\label{prop:chi}
Let ${\bf x}$ be a citation vector and let ${\bf x}^{[k]}$ be a citation vector obtained by adding a single citation to ${\bf x}$.
Then $\chi({\bf x}^{[k]}) \le \chi({\bf x}) +1$.
\end{proposition}
\smallskip

\begin{proof}

We recall that ${\bf x}^{[k]}$ is obtained from ${\bf x}$ by adding a single citation to publication $k$, thereby increasing its citation count to $x_k +1$.
If $rec({\bf x}^{[k]}) = rec({\bf x})$ the result holds trivially. So we may assume that $rec({\bf x}^{[k]}) \not= rec({\bf x})$, in which case $rec({\bf x}^{[k]}) = k (x_k + 1)$ by (\ref{eq:rc}).

\smallskip

Now, since ${\bf x}^{[k]}$ is a citation vector, we must have
\begin{displaymath}
rec({\bf x}) \ge  max \left(k x_k, (k-1) (x_k+1)\right).
\end{displaymath}

Therefore,
\begin{displaymath}
rec({\bf x}) \ge rec({\bf x}^{[k]}) - min (k, x_k+1) \ge rec({\bf x}^{[k]}) - \sqrt{k (x_k+1)},
\end{displaymath}
and thus
\begin{displaymath}
rec({\bf x}^{[k]}) - rec({\bf x}) \le \chi({\bf x}^{[k]}).
\end{displaymath}

It follows that
\begin{displaymath}
\chi({\bf x}^{[k]}) - \chi({\bf x}) \le \frac{\chi({\bf x}^{[k]})}{\chi({\bf x}^{[k]}) + \chi({\bf x})} \le 1.
\end{displaymath}.

\end{proof}
\smallskip

Finally, we note that if it is required that the $\chi$-index be an integer, then the ceiling function, which maps
$\chi({\bf x})$ to the least integer greater than or equal to $\chi({\bf x})$, can be employed.

\section{Concluding remarks}
\label{sec:conc}

Geometric indices, such as the $rec$-index, capture a core set of the publications that represent a researcher's total output. The $\chi$-index (which is equal to the square root of the $rec$-index) can be viewed as a generalisation of the $h$-index, and has the advantage that it allows us to distinguish between more influential and more prolific researchers, depending on whether the height of the largest area rectangle under the citation curve is greater than or less than its width, respectively.

\smallskip

We presented several properties that are satisfied by the $rec$-index and proved, in Theorem~\ref{th:chi2}, that the three properties of {\em monotonicity}, {\em uniform citation} and {\em uniform equivalence} characterise the $rec$-index. While monotonicity is a very natural property for any bibliometric index, uniform citation and uniform equivalence are natural when it is required to include the same number of citations for each publication in the core.

\smallskip

Geometric indices, such as the $rec$-index, give better insight into the relationship between influence (quality) and prolificity (quantity) than indices, such as the $h$-index, that are more constrained in this respect.

\section*{Acknowledgements}

The authors would like to thank the reviewers for their constructive comments, which have helped us to improve the paper.


\newcommand{\etalchar}[1]{$^{#1}$}

\end{document}